\documentclass[11pt]{amsart}
\usepackage{amsfonts,epsf,amsmath}
\usepackage{epsfig}
\usepackage{graphicx}
\usepackage{latexsym}
\usepackage[]{color}
\usepackage{hyperref}
\usepackage{amssymb}

\newtheorem{theorem}{\bf Theorem}[section]

\newtheorem{conjecture}[theorem]{\bf Conjecture}

\def\newpic#1{%
       \def\emline##1##2##3##4##5##6{%
          \put(##1,##2){\special{em:point #1##3}}%
          \put(##4,##5){\special{em:point #1##6}}%
          \special{em:line #1##3,#1##6}}}
\newpic{}

\begin{document}

\title{New results on variants of covering codes in Sierpi\'nski graphs}

\author[S. Gravier]{Sylvain Gravier}
\address{Sylvain Gravier\\F\'ed\'eration de recherche Maths \`{a} Modeler\\Also with Institut Fourier -- UMR 5582 CNRS/Universit\'e Joseph Fourier\\ 100 rue des maths, BP 74, 38402 St Martin d'H\`{e}res, France}
\thanks{Partially supported by ANR/NSC Project GraTel, ANR-09-blan-0373-01 and NSC99-2923-M-110-
001-MY3, 2010--2013}
\email{sylvain.gravier@ujf-grenoble.fr}
 
\author[M. Kov\v{s}e]{Matja\v z Kov\v{s}e}
\address{Matja\v z Kov\v{s}e\\Faculty of Natural Sciences and Mathematics, University of Maribor, Koro\v ska 160, 2000 Maribor, Slovenia\\
Also with the Institute of Mathematics, Physics and Mechanics, Ljubljana.}
\email{matjaz.kovse@gmail.com}
\author[M. Mollard]{Michel Mollard}
\address{Michel Mollard\\F\'ed\'eration de recherche Maths \`{a} Modeler\\Also with Institut Fourier -- UMR 5582 CNRS/Universit\'e Joseph Fourier\\ 100 rue des maths, BP 74, 38402 St Martin d'H\`{e}res, France}
\email{michel.mollard@ujf-grenoble.fr}

\author[J. Moncel]{Julien Moncel}
\address{Julien Moncel\\F\'ed\'eration de recherche Maths \`{a} Modeler\\Also with CNRS -- LAAS Universit\'e de Toulouse\\
UPS, INSA, INP, ISAE ; UT1, UTM, LAAS\\
7 avenue du Colonel Roche\\
31077 Toulouse Cedex 4 (France)}
\thanks{Partially supported by ANR Project IDEA, ANR-08-EMER-007, 2009--2011}
\email{julien.moncel@iut-rodez.fr}

\author[A. Parreau]{Aline Parreau}
\address{Aline Parreau\\F\'ed\'eration de recherche Maths \`{a} Modeler\\Also with Institut Fourier -- UMR 5582 CNRS/Universit\'e Joseph Fourier\\ 100 rue des maths, BP 74, 38402 St Martin d'H\`{e}res, France}
\email{aline.parreau@ujf-grenoble.fr}

\date{}

\maketitle

\begin{abstract}
In this paper we study identifying codes, locating-dominating codes, and total-dominating codes in Sierpi\'nski graphs. We compute the minimum size of such codes in Sierpi\'nski graphs.
\end{abstract}

\noindent



\baselineskip16pt

\section{Introduction and notations}

\subsection{Sierpi\'nski graphs}

Motivated by topological studies 
from~\cite{lipe-92, milu-92} graphs $S(n,k)$ have been introduced in~\cite{klmi-97} and  
named Sierpi\'nski graphs in~\cite{klmi-02}.  In the book~\cite{lipe-09} these graphs are called Klav\v zar-Milutinovi\'c graphs. For $n\geq 2$ and $k\geq 3$, the graph $S(n,k)$ is defined on the vertex
set $\{0, 1, 2, \ldots, k-1\}^n$ (see Figure~\ref{fig:1} for $S(3,3)$ and $S(2,4)$),  where two different vertices $(i_1, i_2,
\ldots, i_n)$ and $(j_1, j_2, \ldots, j_n)$ are adjacent if and
only if there exists an index $h$ in $\{1, 2, \ldots, n\}$ such that
\begin{itemize}
\item[(i)]   $i_t = j_t$, for $t=1,\ldots,h-1$; 
\item[(ii)]  $i_h \neq j_h$; and 
\item[(iii)] $i_t = j_h$ and $j_t = i_h$ for $t=h+1,\ldots,n$.
\end{itemize}

\begin{figure}[ht!]
\centerline{\includegraphics[width=0.9\textwidth]{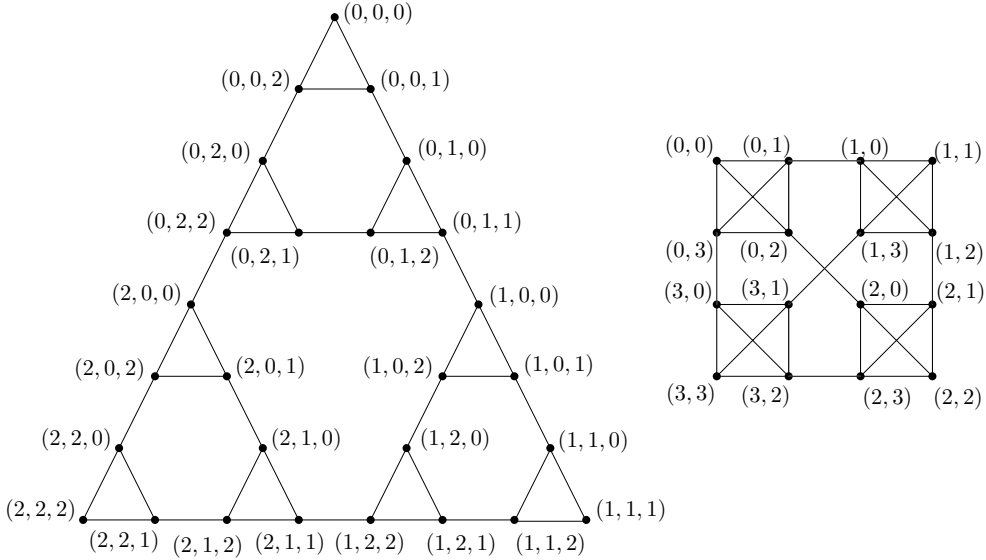}}
\caption{The Sierpi\'nski graphs $S(3,3)$ and $S(2,4)$, together with the 
corresponding vertex labelings. The extreme vertices of $S(3,3)$ are $(0,0,0)$, $(1,1,1)$ and $(2,2,2)$, whereas the extreme vertices of $S(2,4)$ are $(0,0)$, $(1,1)$, $(2,2)$, and $(3,3)$. Notice that each inner vertex $u$ has exactly one neighbour outside the $k$-clique $K(u)$, and that the set of edges $\{uv \mbox{ edge of } S(n,k) \mid K(u)\neq K(v)\}$ is a matching.}\label{fig:1}
\end{figure}

A vertex of the form $(i,i\ldots, i)$
of $S(n,k)$ is called an {\em extreme vertex\/},  all other vertices of $S(n,k)$
are called {\em inner} vertices. The extreme vertices of $S(n,k)$ are of degree $k-1$ while the degree of the inner vertices is $k$. Note also that there are exactly $k$ extreme vertices in $S(n,k)$, and that $S(n,k)$ has $k^n$ vertices. Let $u=( i_1,i_2,\ldots, i_n)$ be an arbitrary vertex of $S(n,k)$. We denote by $K(u)$ the $k$-clique induced by vertices of the form $(i_1,i_2,\ldots, i_{n-1}, j)$, $1\leq j \leq k$. Notice that the neighbourhood of an extreme vertex $u$ is $K(u)\setminus\{u\}$, and that an inner vertex $u$ has only one neighbour that does not belong to $K(u)$. Moreover, the set of edges $\{uv \mbox{ edge of } S(n,k) \mid K(u)\neq K(v)\}$ is a matching. These properties will be extensively used in the sequel.

Notice that the Sierpi\'nski graph $S(n,k)$ can be constructed inductively as follows. We start the construction with a $k$-clique, that can be seen as the Sierpi\'nski graph $S(1,k)$. To construct $S(2,k)$, one has to make $k$ copies of the $k$-clique, that will be connected to each other with an edge set in 1-to-1 correspondence with the edges of a $k$-clique. By repeating this procedure, we can then recursively construct $S(n,k)$ by connecting $k$ copies of $S(n-1,k)$ with a set of $\frac{k(k-1)}{2}$ edges (see Figure~\ref{fig:inductive_construction}). 

\begin{figure}[ht!]
\centerline{\includegraphics[width=0.9\textwidth]{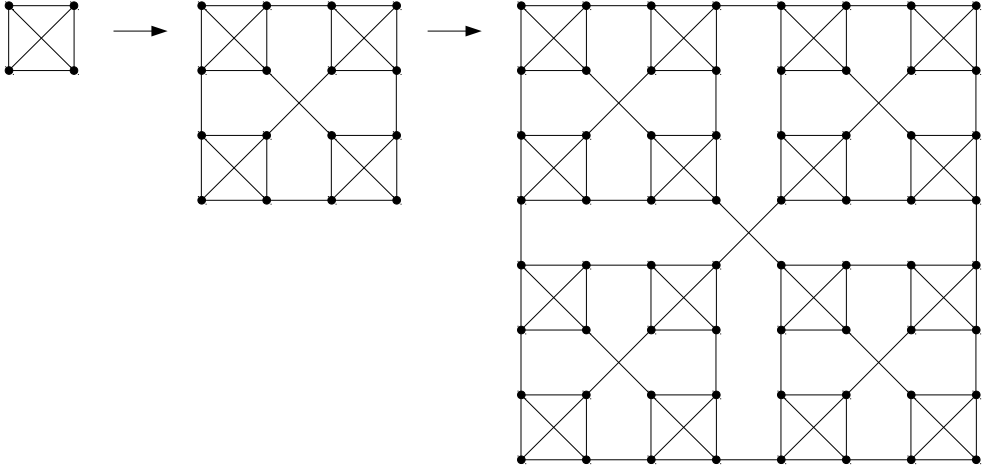}}
\caption{Illustration for $k=4$ of the recursive procedure that enables to construct a Sierpi\'nski graph $S(n,k)$ from $k$ copies of $S(n-1,k)$. From left to right are represented the starting graph $K_4$, then $S(2,4)$ and $S(3,4)$.}\label{fig:inductive_construction}
\end{figure}

Sierpi\'nski graphs $S(n,3)$ coincide with the so-called Tower of Hanoi graphs, which are studied in details in the forthcoming book \cite{book}. Another family of graphs related to Sierpi\'nski graphs are the so-called Sierpi\'nski gasket graphs, where all edges which are themselves maximal cliques are contracted. Some more variations of graphs that can be obtained from Sierpi\'nski graphs have been considered in the papers~\cite{jako-10, klmo-05, Godbole}.  Recently generalized Sierpi\'{n}ski graphs have been introduced in~\cite{grko-10}. They are defined in a similar way as Sierpi\'nski graphs, the difference lying in the fact that the starting graph can be any graph and not just a clique.

Extending the result about the existence of 1-perfect codes in the Tower of Hanoi
graphs from \cite{cune-99} and \cite{line-98}, it has been shown in \cite{klmi-02} that 1-perfect codes exist for all graphs $S(n,k)$. More precisely it is shown in \cite{klmi-02} that
$$\gamma(S(n,k))= 
\begin{cases} 
k \cdot \frac{k^{n-1}+1}{k+1},  & \mbox{if }n\mbox{ is even} \\
\frac{k^{n}+1}{k+1}, & \mbox{if }n\mbox{ is odd} 
\end{cases}$$ where $\gamma(G)$ denotes the domination number of a graph $G$. This result has been further generalized in \cite{begr-10}, where all existing $(a,b)$-codes of Sierpi\'nski graphs have been characterized. 
Sierpi\'nski graphs  have also been studied for $L(2,1)$-labelings~\cite{grkl-05}, 
crossing numbers~\cite{klmo-05}, and different types of colorings~\cite{klav-08,klja-09}. 

In this note we compute the minimum size of identifying codes, locating-dominating codes, and total-dominating codes in Sierpi\'nski graphs.

\subsection{Codes in graphs}

In a simple undirected graph $G$, let $d(u,v)$ denote the minimum number of edges of a path having $u$ and $v$ as extremities. Let $C$ be a subset of vertices of $G$, we call $C$ a \emph{code} of graph $G$. For a vertex $u$ of $G$, let us denote $I(u,C) = \{v\in C \mid d(u,v)\leq 1\}$. We say that $C$ {\it covers} a vertex $u$ if we have $I(u,C)\neq \emptyset$. In other words, $C$ covers $u$ if and only if we either have $u\in C$, or there exists $v\in C$ such that $u$ and $v$ are neighbours. The code $C$ is said to {\it totally cover} a vertex $u$ if we have $I(u,C)\setminus \{u\}\neq \emptyset$. We say that $C$ {\it separates} two distinct vertices $u,v$ if $I(u,C) \neq I(v,C)$.

The code $C$ is a {\it covering code} (or a dominating code) of $G$ if $C$ covers all the vertices of $G$. It is a {\it total-dominating code} of $G$ if it totally covers all the vertices of $G$. It is an {\it identifying code} if it is a covering code of $G$ that separates all pairs of distinct vertices of $G$. It is a {\it locating-dominating code} if it is a covering code of $G$ that separates all pairs of distinct vertices of $G$, where neither of the vertices of a pair belongs to $C$. This terminology is standard, see~\cite{coho-97} for more about covering codes and \cite{lobs} for a bibliography of problems about identifying codes and locating-dominating codes.

Many papers in the recent literature deal with these codes considered in regular structures, such as hypercubes \cite{cchl, exoo, m06}, lattices \cite{bglm11, bl05, chhl, hl08}, cycles \cite{BCHL04, CLM11, gms06, JL11, rr, xu}, and Cartesian products of graphs \cite{KnsquareKn, mollard}. In the present paper we propose to address the computation of the minimum cardinality of such codes in Sierpi\'nski graphs.

\section{Identifying codes in Sierpi\'nski graphs}

\begin{theorem}\label{thm:idcodes}
The minimum cardinality of an identifying code in a Sierpi\'nski graph $S(n,k)$ is $k^{n-1}(k-1)$.
\end{theorem}

\begin{proof}
We shall first prove that for any identifying code $C$ of $S(n,k)$, we have $|C| \geq k^{n-1}(k-1)$. For each inner vertex $u$ of $S(n,k)$, let $m(u)$ denote the unique neighbour of $u$ which does not belong to $K(u)$. For each $k$-clique $K$ of $S(n,k)$, let us consider the set $M(K)=\{m(u) \mid u \text{ is an inner vertex of } K\}$. Clearly, the $M(K)$'s are all disjoint, and the number of such sets is $k^{n-1}$. Now, by way of contradiction, assume that there exist two inner vertices $u, v$ of a given $k$-clique $K$ such that $m(u)$ and $m(v)$ do not belong to $C$. Then $u$ and $v$ would not be separated, a contradiction. Moreover, if $K$ contains an extreme vertex, all the vertices of $M(K)$ must be in $C$ to separate the extreme vertex from the other vertices of
  $K$. Hence, for each $k$-clique $K$, there are at least $k-1$ vertices of $M(K)$ that belong to $C$. Since there are $k^{n-1}$ sets $M(K)$, and they are all disjoint, this implies that the cardinality of $C$ is at least $k^{n-1}(k-1)$.

Now, observe that if $n=2$, then the set of inner vertices of $S(2,k)$ is an identifying code of $S(2,k)$, of cardinality $k(k-1)$. Hence the minimum cardinality of an identifying code of $S(2,k)$ is $k(k-1)$. Now let us consider the general case $S(n,k)$, with $n\geq 2$. Let $C$ be the subset of vertices of $S(n,k)$ such that, for any subgraph $S$ of $S(n,k)$ which is isomorphic to $S(2,k)$, the vertices of $S$ that belong to $C$ are exactly the inner vertices of $S$. It is easy to see that $C$ is an identifying code of $S(n,k)$, of cardinality $k^{n-1}(k-1)$.
\end{proof}

\section{Locating-dominating codes in Sierpi\'nski graphs}

\begin{theorem}\label{thm:licodes}
The minimum cardinality of a locating-dominating code in a Sierpi\'nski graph $S(n,k)$ is $\frac{k^{n-1}(k-1)}{2}$.
\end{theorem} 
 
\begin{proof}
Let us first show that if $C$ is a locating-dominating code of $S(n,k)$, then $|C| \geq \frac{k^{n-1}(k-1)}{2}$. Let $C$ be a locating-dominating code. 
We call an edge $uv$ a {\em crossing edge} if it is a maximal clique itself, that is to say $K(u)\neq K(v)$.
The set of crossing edges is a matching and each inner vertex belongs to exactly one crossing edge. Consider now the following partition $\mathcal P$ of the vertices of $S(n,k)$: Each extreme vertex is itself a part of $\mathcal P$ and each pair of vertices $\{u,v\}$ forming a crossing edge is a part of $\mathcal P$.
For each vertex $u$ of $S(n,k)$, let $P(u)$ be the set of $\mathcal P$ containing $u$.

By way of contradiction, assume that there exist two vertices $u, v$ of a given $k$-clique $K$ such that $P(u)\cup P(v)$ contains no element of $C$. Then $u$ and $v$ are not in the code and not separated, a contradiction. Hence, for any $k$-clique $K$, at least $k-1$ sets among the $k$ sets of $\mathcal P$ containing one vertex of $K$ must contain a vertex of $C$. A set of $\mathcal P$ can contribute to at most two different $k$-cliques and there are $k^{n-1}$ $k$-cliques. Therefore, $|C| \geq \frac{k^{n-1}(k-1)}{2}$.

Now, let us consider any code $C$ such that for each $k$-clique $K$, $k-1$ of the sets of $\mathcal P$ containing a vertex of $K$ contain at least one vertex of $C$. It is easy to see that if this code $C$ is such that each $k$-clique of $S(n,k)$ contains at least one vertex of $C$, then $C$ is a locating-dominating code of $S(n,k)$. One can construct such a code $C$ of cardinality $\frac{k^{n-1}(k-1)}{2}$ with the following process.
\begin{itemize}
\item For $n=2$, consider any code $C$ with exactly one vertex on each crossing edge and at least one vertex on each clique. To construct such a code, one can for example take an hamiltonian cycle of $K_k$, and for the corresponding crossing edge (there is a natural bijection between the crossing edges of $S(2,k)$ and the edges of a $k$-clique), take every second vertex.
Clearly, this code has size $\frac{k(k-1)}{2}$.
\item For $n\geq 3$, on each subgraph $S$ of $S(n,k)$ which is isomorphic to $S(2,k)$, we put in the code $C$ the vertices of $S$ corresponding to the previous code of $S(2,k)$. There are $k^{n-2}$ such (disjoint) subgraphs $S$, hence, $C$ has size $\frac{k^{n-1}(k-1)}{2}$.
\end{itemize}

\end{proof}

 \section{Total-dominating codes in Sierpi\'nski graphs}

 \begin{theorem}\label{thm:tdcodes} 
 Let $\gamma_t(G)$ denote the minimum cardinality of a total-dominating code in a graph $G$. Then we have $\gamma_t(S(n,k))=\begin{cases} 
k^{n-1},  & \mbox{if }k \mbox{ is even} \\
k^{n-1} + 1, & \mbox{if }k\mbox{ is odd} 
\end{cases}\,.$
\end{theorem}

\begin{proof} 
We first show that in any subgraph isomorphic to $S(2,k)$, at least $k$ vertices must belong to a total-dominating code. Let $S$ be a subgraph isomorphic to $S(2,k)$ and let $C$ be a total-dominating code of $S(n,k)$. If all the $k$-cliques of $S$ contain at least one vertex of $C$, then we are done. Hence, let us assume that there is a $k$-clique $K$ of $S$ containing no vertex of $C$. This implies that $K$ does not contain any extreme vertex. Then each vertex $u$ of $K$ has exactly one neighbour $m(u)$ that is not in $K$, and must therefore be in $C$. In addition, at least one neighbour of $m(u)$, which is in $K(m(u))$, must also be in $C$. 
Each other $k$-clique of $S$ must then contain at least two vertices of $C$. Hence $S$ contains at least $2(k-1)\geq k$ vertices of $C$, and we are done. This proves that the cardinality of a total-dominating code is greater than or equal to $k^{n-1}$.

Now we distinguish two cases. First let $k$ be an even number. Let $C$ bet the set of vertices of $S(2,k)$ induced by a perfect matching of $K_k$, such that vertices $(i,j)$ and $(j,i)$ of $S(2,k)$ belong to $C$ if and only if edge $\{i,j\}$ of $K_k$ belongs to the perfect matching. Set $C$ covers all vertices of $S(2,k)$, and since any vertex from $C$ has also a neighbour from another clique which is also in $C$, then the set $C$ is a total-dominating code. Hence $C$ is a minimum total-dominating code of $S(2,k)$. Using the same arguments for each copy of $S(2,k)$ in $S(n,k)$ it follows that $\gamma_t(S(n,k))= k^{n-1}$, when $k$ is even. Note that since each vertex of a total-dominating code $C$ must be a neighbour of another vertex of $C$, then all minimum total-dominating codes can be constructed in such a way.

When $k$ is odd, it is impossible to find a total-dominating code of $k^{n-1}$ vertices intersecting each $k$-clique exactly once. Indeed, in this case, each vertex of $C$ would be neighbour to exactly one another vertex of $C$. Hence there would be a perfect matching running on an odd number of vertices, which is impossible. Hence, when $k$ is odd, the cardinality of a total-dominating code is greater than or equal to $k^{n-1}+1$.

Let us first consider the case $n=2$. To construct a minimum total-dominating code, we can similarly take a maximum matching of $K_k$, and build set $C$ in the same way as in the previous case. In this case one $k$-clique of $S(2,k)$ is not dominated by $C$ and we must add two of its vertices to $C$ to get a minimum total-dominating code. This gives altogether one additional vertex in $C$ compared to the size of $C$ in the even case. 

For $n=3$ we choose in each copy of $S(2,k)$ inside $S(3,k)$, a subset of vertices induced by a maximum matching of $K_k$ and put them in the set $C$. For a $k$-clique that is not dominated, we choose the extreme vertex of $S(2,k)$ and put it in set $C$. We can choose $k$-cliques belonging to different copies of $S(2,k)$ so that they induce a maximum matching in $K_k$, and choose corresponding extreme vertices and put them into $C$. Similarly as in the case $n=2$, one $k$-clique remains which has only one vertex in $C$ and therefore we must add one more vertex from this $k$-clique to $C$ in order to get a minimum total-dominating code of $S(3,k)$. By iteratively repeating this procedure we construct a minimum total-dominating code of $S(n,k)$, for any given $n\geq 4$. Therefore it follows that $\gamma_t(S(n,k))= k^{n-1}+1$, when $k$ is odd.
\end{proof}

\section{Remarks}
 
Foucaud et al. \cite{foklkora-sub} propose the following conjecture about an upper bound of the cardinality of a minimum identifying code of a given graph in terms of its number of vertices and maximum degree. We recall that two adjacent vertices are said to be \emph{twins} if they have the same neighbourhood. A graph having no twins is said twin-free. A graph admits an identifying code if and only if it is twin-free.

\begin{conjecture}[\cite{foklkora-sub}]
For every connected twin-free graph $G$ of maximum degree $\Delta\ge~3$, the minimum cardinality of an identifying code of $G$ is less than or equal to $\Big\lceil|V(G)|-\frac{|V(G)|}{\Delta(G)}\Big\rceil$.
\end{conjecture}

From the results of this paper it follows that  Sierpi\'nski graphs attain this bound. It might be interesting to consider another families of graphs attaining this bound that might be constructed in a similar manner as Sierpi\'nski graphs or to provide a counterexample to the given conjecture.

\end{document}